\documentclass[runningheads]{llncs}
\usepackage{amsmath}

\usepackage{tikz}
\usepackage{pgfplots}
\usepackage{xcolor}
\usetikzlibrary{fpu,shapes,patterns}
\usepackage{color}
\usepackage[pdfencoding=auto]{hyperref}
\usepackage{cleveref}
\usepackage{subcaption}
\usepackage{ifthen}
\usepackage[misc]{ifsym}

\hypersetup{
    colorlinks=true,
    linkcolor=blue,
    urlcolor=blue,
    linktoc=all,
    citecolor=blue
}

\def\isanonymized{0}  

\newcommand{\conditionalAnonymize}[1]{%
    \ifnum\isanonymized=1
        Anonymized%
    \else
        #1%
    \fi
}

\definecolor{okabe1}{HTML}{000000}
\definecolor{okabe2}{HTML}{E69F00}
\definecolor{okabe3}{HTML}{56B4E9}
\definecolor{okabe4}{HTML}{009E73}
\definecolor{okabe5}{HTML}{F0E442}
\definecolor{okabe6}{HTML}{0072B2}
\definecolor{okabe7}{HTML}{D55E00}
\definecolor{okabe8}{HTML}{CC79A7}

\newcommand{\bem}[1]{\textbf{\emph{#1}}}

\newcommand{\Crefpart}[2]{%
	\nameCref{#1}~\hyperref[#2]{\labelcref*{#1}.\ref*{#2}}%
}
\pgfplotsset{compat=1.18}

\begin{document}
\title{Highway Preferential Attachment Models for Geographic Routing}
\author{\conditionalAnonymize{Ofek Gila\textsuperscript{(\Letter)}} \and
\conditionalAnonymize{Evrim Ozel\textsuperscript{(\Letter)}} \and
\conditionalAnonymize{Michael Goodrich\textsuperscript{(\Letter)}}}
\authorrunning{\conditionalAnonymize{O. Gila} et al.}
\institute{\conditionalAnonymize{University of California, Irvine CA 92617, USA \email{\{ogila,eozel,goodrich\}@uci.edu}}}

\maketitle

\begin{abstract}

	In the 1960s, the world-renowned social psychologist Stanley Milgram
	conducted experiments that showed that not only do there exist ``short
	chains'' of acquaintances between any two arbitrary people, but that these
	arbitrary strangers are able to find these short chains. This phenomenon,
	known as the \emph{small-world phenomenon}, is explained in part by any model
	that has a low diameter, such as the Barab\'asi and Albert's
	\emph{preferential attachment} model, but these models do not display the
	same efficient routing that Milgram's experiments showed. In the year 2000,
	Kleinberg proposed a model with an efficient $\mathcal{O}(\log^2{n})$ greedy
	routing algorithm. In 2004, Martel and Nguyen showed 
    that Kleinberg's analysis was tight, while also showing that Kleinberg's model had an expected
	diameter of only $\Theta(\log{n})$---a much smaller value than the greedy
	routing algorithm's path lengths. In 2022, Goodrich and Ozel proposed the
	\emph{neighborhood preferential attachment} model (NPA), combining elements from
	Barab\'asi and Albert's model with Kleinberg's model, and experimentally
	showed that the resulting model outperformed Kleinberg's greedy routing
	performance on U.S. road networks. While they displayed impressive empirical
	results, they did not provide any theoretical analysis of their model. In
	this paper, we first provide a theoretical analysis of a generalization of
	Kleinberg's original model and show that it 
    can achieve expected $\mathcal{O}(\log{n})$ 
    routing, a much better result than Kleinberg's model. We then propose a new model, \emph{windowed NPA}, that is
	similar to the neighborhood preferential attachment model but has provable theoretical
	guarantees w.h.p. We show that this model is able to achieve
	$\mathcal{O}(\log^{1 + \epsilon}{n})$ greedy routing for any $\epsilon > 0$.

\end{abstract}



\keywords{small worlds, social networks, random graphs}

\section{Introduction}

Stanley Milgram, a social psychologist, popularized the concept of the
\bem{small-world phenomenon} through two groundbreaking experiments in the 1960s
\cite{milgramstanleySmallWorldProblem1967,traversExperimentalStudySmall1969}. In
these experiments, Milgram determined that the median number of hops from a
random volunteer in Nebraska and Boston to a stockbroker in Boston was six,
thereby giving rise to the expression ``six degrees of separation''.

A common and well-studied method for modeling real-world social networks is the
\bem{preferential attachment} model, popularized by Barab\'asi and Albert in
1999 \cite{barabasiEmergenceScalingRandom1999}. In this model, nodes are added
to the graph one at a time, and each node is connected to $m$ other nodes with
probability proportional to their degree. Put simply, in this model, nodes with
a greater degree are more likely to obtain an even greater degree, in what is
commonly referred to as a ``rich-get-richer'' process. Such a process leads to
power law degree distributions, meaning that the number of nodes with degree $k$
is proportional to $k^{-\alpha}$ for some constant $\alpha > 1$. In 2009,
Dommers, Hofstad, and Hooghiemstra showed that the diameter of the preferential
attachment model is $\Omega(\log{n})$ when the power law exponent $\alpha > 3$,
and $\Omega(\log{\log{n}})$ when $\alpha \in (2, 3)$
\cite{dommersDiametersPreferentialAttachment2010}. While such preferential
attachment models indeed display small diameters, therefore explaining how these
short paths \textit{exist}, they do not explain how these paths are
\textit{found}. In other words, individual nodes in these models, using only
local information, cannot find short paths to other nodes, unlike in Milgram's
experiments.

In 2003 Dodds, Muhamad, and Watts conducted an experiment similar to Milgram's
using email, with more than 60,000 volunteers and 18 targets in 13 countries.
This experiment determined that the average number of hops was around five if
the target was in the same country and seven if the target was in a different
country, largely in line with Milgram's results. This experiment
asked participants the reasons for picking their next particular acquaintance,
finding that, especially during the early stages of routing, geographical
proximity was the dominant factor \cite{doddsExperimentalStudySearch2003}. This
result suggests that realistic models aiming to explain the small-world
phenomenon should incorporate geographical information.

\subsection{Kleinberg's Model}

In 2000, Jon Kleinberg proposed a famous model that, while not incorporating
true geographical information, does consider a notion of geographic distance
by placing nodes on an $n \times n$ grid. Kleinberg's model connects nodes using
two types of connections---\bem{local connections}, in which nodes are connected
to all neighbors within a fixed lattice distance, and \bem{long-range
connections}, in which nodes are connected to random nodes in the graph.
Importantly, these long-range connections are chosen with distance in mind,
namely that closer nodes are picked more often as long-range connections than
farther nodes. Specifically, each node $u$ picks long-range connection $v$ with
probability proportional to $d(u, v)^{-s}$, where $d(u, v)$ is the lattice
distance between $u$ and $v$ and $s$ is the clustering exponent. This model
mimics how individuals in a social network are more likely to know people who
are geographically closer to them, but also have a small probability of knowing
people who are farther away. Kleinberg showed that, for $s = 2$, a greedy
routing algorithm can find paths of length $\mathcal{O}(\log^2{n})$ with high
probability (w.h.p.), and that this is optimal for any $s$\footnote{for 2-d
grids.} \cite{kleinbergSmallworldPhenomenonAlgorithmic2000a}. In 2004, Martel and
Nguyen proved tight bounds of expected $\Theta(\log^2{n})$ hops for greedy
routing, and of expected diameter of $\Theta(\log{n})$---highlighting the large
discrepancy between the two \cite{martelAnalyzingKleinbergOther2004a}. We are
not aware of any other work that achieves an asymptotically better expected
number of greedy routing hops using a constant average node degree and using
only a constant average amount of local information per node.

\subsection{The Neighborhood Preferential Attachment Model}

In 2022, Goodrich and Ozel proposed a new model that combines the preferential
attachment model with Kleinberg's model, which they call the \bem{neighborhood
preferential attachment} model \cite{goodrichModelingSmallworldPhenomenon2022}.
In this model, as in the Barab\'asi-Albert model, nodes are added to the graph
one at a time, but instead of connecting to nodes solely based on their degree
as in the preferential attachment models, they also take into account the
distance between the nodes, as in Kleinberg's model. Specifically, each node $u$
picks a node $v$ with probability proportional to $\deg(v) / d(u, v)^s$, where
$\deg(v)$ is the current degree of vertex $v$. Furthermore, Goodrich and Ozel
expanded all three models (Barab\'asi-Albert, Kleinberg, and their own) to work
with underlying distances defined by a road network rather than a grid. In their
work, they conducted rigorous experiments on U.S.A. road networks and showed
that their model is able to outperform both the
constituent
models in terms of
average greedy routing hops between randomly chosen pairs of nodes.
In their paper, they describe how road networks serve as good proxies for social networks since the density of road infrastructure is correlated with population density.
Their model was, at the time, the only randomized model to not only capture a proxy for the position of nodes in a social network, but also the power law distribution of node degrees that is widely common social networks.
These two facts allowed this model to be the first randomized model able to reproduce results from Stanley Milgram's original small-worlds social experiment using a small average degree (only of around 30). 
However,
importantly, they did not prove any theoretical bounds on their model. Our paper
can be seen as a theoretical complement to their work, as we prove high
probability bounds on the average greedy routing path length of a grid version
of a very similar model, showing that it is far better than the
$\Theta(\log^2{n})$ bound of Kleinberg's model.

\subsection{Our Results}

As stated before, our main goal for this paper was to provide theoretical
results for the work of Goodrich and Ozel, or more generally, for preferential
attachment variations of Kleinberg's model. In this paper, we propose three new
models, each combining aspects of both Kleinberg's model and the preferential
attachment model. We prove that, for grid networks, each of our networks are
able to asymptotically outperform Kleinberg's original model in terms of average
greedy routing path length, while maintaining an expected constant average node
degree and using a constant average amount of local information per node.

We note that greedy routing can be improved by relaxing either of these two
constraints. For example, if we allow nodes in the Kleinberg model to have
access to more local information, we can improve greedy routing to
$\mathcal{O}(\log^{3/2}{n})$. Similarly, if we allow nodes to have a higher,
$\mathcal{O}(\log{n})$, average degree, then we can improve greedy routing to
$\mathcal{O}(\log{n})$ hops \cite{martelAnalyzingKleinbergOther2004a}. The
latter of these two relaxations reveals that greedy routing can be greatly
improved by getting to---and staying on---high degree nodes. With this in mind,
we consider a node \bem{highway}---a set of interconnected nodes that each have
higher than average degrees. Our first two models introduce a parameter $k$ that
controls both the size of the highway and the degree of nodes on the highway.
Specifically, the degree of nodes on the highway is proportional to $k$ while
the number of highway nodes is inversely proportional to $k$, such that
the average degree of the entire graph is constant.

Our first model, the \bem{Kleinberg highway model} (KH), works by embedding a
Kleinberg grid within an $n \times n$ grid, such that there are $n^2 / k$ nodes
on the highway. Each of the nodes on the highway grid only chooses long-range
connections to other nodes on the highway grid, while local connections are
still made to all neighbors within a fixed lattice distance as in the original
Kleinberg model. Our second model, the \bem{randomized highway model} (RH), is
similar to the first, but instead of embedding a perfect Kleinberg grid inside
the original graph, nodes are chosen uniformly at random to be on the highway
grid. More specifically, each node has probability $1/k$ to become a highway
node, leading to an expected $\Theta(n^2/k)$ highway nodes w.h.p. Both of these
generalizations reduce to the original Kleinberg model when $k = 1$, that is
when every node is a highway node, and adds a constant number of long-range
connections per node. Importantly, both models reach a global minimum of
$\mathcal{O}(\log{n})$ hops when $k = \Theta(\log{n})$, a much better result
than Kleinberg's (see \Cref{fig:paramk}).


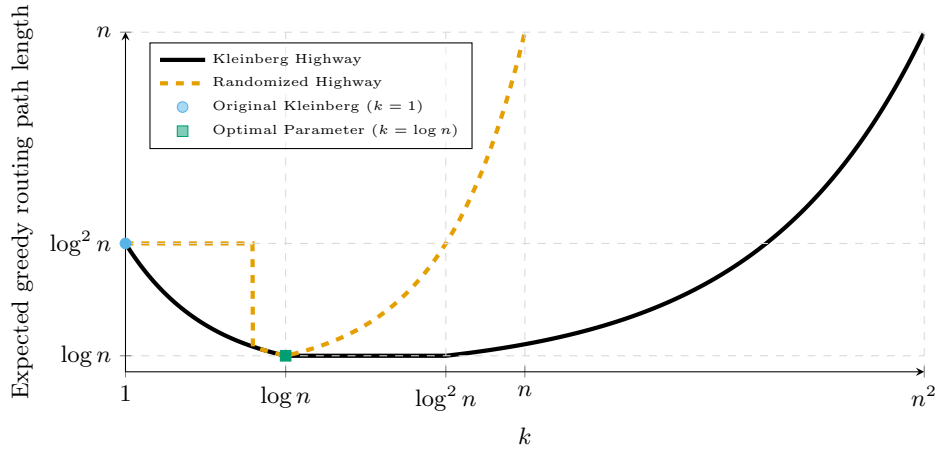
\begin{figure}
    \centering
    \begin{tikzpicture}
    \pgfmathsetmacro{\n}{128}

    \pgfkeys{/pgf/fpu=true}
    \pgfmathparse{\n * \n}
    \edef\nsq{\pgfmathresult}
    \pgfkeys{/pgf/fpu=false}
    
    \pgfmathsetmacro{\logn}{log2(\n)}
    \pgfmathsetmacro{\lognsq}{\logn * \logn}
    \pgfmathsetmacro{\twologn}{2 * \logn}
    \pgfmathsetmacro{\sqrtn}{sqrt(\n)}
    \pgfmathsetmacro{\q}{\logn / log2(log2(log2(\n)))}

    \begin{semilogxaxis}[
        xlabel=$k$,
        ylabel={Expected greedy routing path length},
        xmin=1,
        xmax=\nsq,
        ymin=1,
        ymax=\n,
        xtick={1,\logn,\lognsq,\n,\nsq},
        xticklabels={$1$, $\log{n}$, $\log^2{n}$, $n$, $n^2$},
        ytick={\logn,\lognsq,\n},
        yticklabels={$\log{n}$, $\log^2{n}$, $n$},
        legend pos=north west,
        legend cell align={left},
        legend style={font=\tiny,fill=white, fill opacity=0.5,draw=black,text opacity=1},
        width=\columnwidth,
        height=0.5\columnwidth,
        grid=major,
        grid style={dashed,gray!30},
        axis lines=left,
        enlargelimits=false,
        clip=false,
        axis on top,
        scale=1.0,
        legend entries={Kleinberg Highway, Randomized Highway, Original Kleinberg ($k = 1$), Optimal Parameter ($k = \log{n}$)}
    ]
        \addplot[domain=1:\nsq, color=okabe1, ultra thick,samples=1000] {
            ifthenelse(x<\logn, \lognsq/x, max(sqrt(x), \logn))};

        \addplot[domain=1:\n, color=okabe2, ultra thick, dashed, samples=1000] {
            ifthenelse(x<\q, \lognsq, ifthenelse(x<\logn, \lognsq/x, x))};

        \addplot[only marks, mark=*, mark size=2pt, text mark as node=true, color=okabe3] coordinates {(1,\lognsq)};
        \addplot[only marks, mark=square*, mark size=2pt, text mark as node=true, color=okabe4] coordinates {(\logn,\logn)};
    \end{semilogxaxis}
\end{tikzpicture}
    \vspace*{-0.5cm}
    \caption{ \label{fig:paramk}The average greedy routing path length of the
        Kleinberg highway model for different values of parameter $k$. }
\end{figure}

Our final model is the \bem{windowed neighborhood preferential attachment model}
(windowed NPA), which like Goodrich and Ozel's neighborhood preferential
attachment model (NPA), is based on both Kleinberg's model and the preferential
attachment model. There are two main differences between the models. First, in
the NPA model, the power law degree distribution naturally arises from the
rich-get-richer selection property when adding new edges. In contrast, in our
model, the power law degree distribution is strictly enforced, with each node
picking a popularity $k$ with probability $\propto k^{-\alpha}$. Each node node
then adds a number of long-range connections proportional to its popularity. In
order to maintain a constant average degree, the power law exponent $\alpha$
must be greater than 2, so $\alpha \geq 2 + \epsilon$ for any $\epsilon > 0$.
The second main difference is that instead of there existing a probability of
any two nodes being connected, in the windowed NPA model, nodes are only
connected to other nodes within a constant factor of their popularity. The idea
being that a residential street is more likely to connect to an alley, another
residential street, or an arterial road, than it is to connect directly to a
highway. This constant factor is controlled by a parameter $A$, and any node $u$
with popularity $k_u$ can only have long-range connections to nodes with
popularity $k_v$ such that $k_v \in [k_u / A, k_u \cdot A]$. We prove that for
any arbitrarily small $\epsilon > 0$, the average greedy routing path length of
the windowed NPA model is $\mathcal{O}(\log^{1 + \epsilon}{n})$
w.h.p.\footnote{We proved this for a slightly modified greedy routing
algorithm.} While this result only holds for grid networks, we provide
experimental results of our new model on both grid and road networks, showing
that the windowed NPA model is able to outperform Kleinberg's model on both
types of networks.

\section{Preliminaries}


As stated before, for the theoretical analysis, we will be using an $n \times n$
grid, such that the total number of nodes $|V| = n^2$. For simplicity, we will
assume that our grid has wrap-around edges, as is common when analyzing grid
networks \cite{martelAnalyzingKleinbergOther2004a}, although our results can be
extended to non-wrap-around grids. Let $d(u, v)$ be defined as the lattice
distance between two nodes $u$ and $v$ in the grid, i.e. $d(u, v) =
\min(\delta_x, n - \delta_x) + \min(\delta_y, n - \delta_y)$, where $\delta_x$
and $\delta_y$ are the absolute differences in the $x$ and $y$ coordinates of
$u$ and $v$, respectively. Let $B_d(u)$ denote the set of nodes within lattice
distance $d$ from $u$. All three models have the notion of \bem{local
connections} and \bem{long-range connections}. Without loss of generality, we
will only consider the case where we only add immediately adjacent local
connections, that is, each node is only connected to the four nodes directly
above, below, to the left, and to the right of it. Equivalently, we can say that
each node is connected to all other nodes in $B_1(u)$, as in the case when $p =
1$ in Kleinberg's original model. In this paper, when we refer to a node's
degree $\deg(u)$, we will be referring to the number of long-range
connections from $u$.


\section{Kleinberg Highway}\label{sec:kh}

As stated before, Kleinberg's model is defined on a graph $\mathcal{G}$
comprising of an $n \times n$ grid where each node $u$ adds local connections to
all nodes in $B_P(u)$ (all nodes within lattice distance $P$ of $u$), and $Q$
long-range connections to other nodes. The probability of adding a long-range
connection to node $v$ is proportional to $d(u, v)^{-r}$. In our model, we will
set $P$ to 1 w.l.g., and we will set $r = 2$, as this is the value that
Kleinberg showed was optimal for 2-dimensional grids, and Goodrich and Ozel
hypothesized could be optimal for road networks
\cite{kleinbergSmallworldPhenomenonAlgorithmic2000a,goodrichModelingSmallworldPhenomenon2022}. Furthermore, in our model, we will
define a subgraph $\mathcal{G}_H$, known as the \bem{highway}, which for this
model is an $n_H \times n_H$ evenly spaced grid in $\mathcal{G}$. We introduce a
new parameter $k$ in the range of $1 \leq k \leq n^2$, where $1/k$ of the nodes
are designated as \bem{highway nodes}, meaning that $n_H$ is equal to
$n/\sqrt{k}$ (which for simplicity we assume is a whole number). Now, we
introduce two forms of local connections, the first connects all nodes in the
entire graph $\mathcal{G}$ to their neighbors, and the second connects all nodes
in the highway subgraph $\mathcal{G}_H$ to their highway neighbors. Finally, and
importantly, only highway nodes are able to add long-range connections, and
these long-range connections are \textit{directed} edges added only \textit{to
other highway nodes} (see \Cref{fig:kh}). Since there are fewer highway nodes,
we are able to add proportionally more long-range connections per node to
maintain the same constant average degree $Q$. In particular, each highway node
is able to add $Q \times k$ long-range connections, where $Q$, as in the
original Kleinberg model, represents the average highway degree. Put simply,
$\mathcal{G}_H$ is a Kleinberg graph with Kleinberg parameters: $n = n_H =
n/\sqrt{k}$, $p = 1$, $q = Q \times k$, $r = 2$. We call the entire graph
$\mathcal{G}$ the \bem{Kleinberg highway} model.

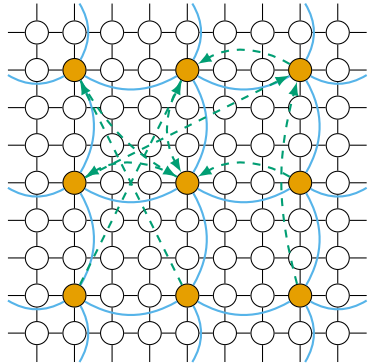
\begin{figure}
	\centering
	\begin{tikzpicture}[
	c/.style={circle, draw=okabe1},
	h/.style={circle, draw=okabe1, fill=okabe2},
	l/.style={draw=okabe3,thick},
	lr/.style={draw=okabe4,thick,dashed,-latex},
]
	\pgfmathsetmacro{\n}{9}
	\pgfmathsetmacro{\k}{3}
	\pgfmathtruncatemacro{\npo}{\n + 1}
	\pgfmathtruncatemacro{\kmo}{\k - 1}
	\pgfmathtruncatemacro{\kmopk}{\kmo + \k}
	\pgfmathtruncatemacro{\nmk}{\n - \k}
	\pgfmathtruncatemacro{\lastk}{\n - mod(\n, \k) - 1}

	\foreach \x in {1,...,\n} {
		\foreach \y in {1,...,\n} {
			\pgfmathtruncatemacro{\isHighway}{(mod(\x,\k)==\kmo && mod(\y,\k)==\kmo) ? 1 : 0}
			\ifnum\isHighway=1
				\node[h] (\x-\y) at (0.5*\x, 0.5*\y) {};
			\else
				\node[c] (\x-\y) at (0.5*\x, 0.5*\y) {};
			\fi
		}
	}

	\foreach \x in {0,...,\npo} {
		\node (\x-0) at (0.5*\x, 0) {};
		\node (\x-\npo) at (0.5*\x, 0.5*\npo) {};
	}

	\foreach \y in {1,...,\n} {
		\node (0-\y) at (0, 0.5*\y) {};
		\node (\npo-\y) at (0.5*\npo, 0.5*\y) {};
	}

	\foreach \x in {1,...,\n} {
		\foreach \y in {1,...,\n} {
			\pgfmathtruncatemacro{\xpo}{\x + 1}
			\pgfmathtruncatemacro{\ypo}{\y + 1}

			\draw (\x-\y) -- (\x-\ypo);
			\draw (\x-\y) -- (\xpo-\y);
		}
	}

	\foreach \i in {1,...,\n} {
		\draw (\i-0) -- (\i-1);
		\draw (0-\i) -- (1-\i);
	}

	\foreach \x in {\kmo,\kmopk,...,\nmk} {
		\pgfmathtruncatemacro{\xpk}{\x + \k}

		\foreach \y in {\kmo,\kmopk,...,\nmk} {
			
			\pgfmathtruncatemacro{\ypk}{\y + \k}

			\draw[l] (\x-\y) to [bend right=30] (\x-\ypk);
			\draw[l] (\x-\y) to [bend right=30] (\xpk-\y);
		}

		\draw[l] (\lastk-\x) to [bend right=30] (\lastk-\xpk);
		\draw[l] (\x-\lastk) to [bend right=30] (\xpk-\lastk);
	}

	\foreach \x in {\kmo,\kmopk,...,\n} {
		\draw[l] (0-\x) to [bend right=30] (\kmo-\x);
		\draw[l] (\x-\lastk) to [bend right=30] (\x-\npo);
		\draw[l] (\lastk-\x) to [bend right=30] (\npo-\x);
		\draw[l] (\x-0) to [bend right=30] (\x-\kmo);
	}

	\draw[lr] (2-2) to (5-8);
	\draw[lr] (2-5) to (8-8);
	\draw[lr] (2-8) to [bend right=15] (5-5);
	\draw[lr] (5-2) to (2-8);
	\draw[lr] (5-5) to [bend right=30] (2-5);
	\draw[lr] (5-8) to [bend right=30] (5-5);
	\draw[lr] (8-2) to [bend left=15] (8-8);
	\draw[lr] (8-5) to [bend right=30] (5-5);
	\draw[lr] (8-8) to [bend right=30] (5-8);
\end{tikzpicture}
	\vspace*{-0.25cm}
	\caption{\label{fig:kh}An example of the Kleinberg highway model with $n =
		9$, $k = 9$, and $Q = 1/9$. The solid black and curved solid blue lines
		represent local connections for the entire grid and for the highway
		grid, respectively. The value of $Q$ was picked such that each highway
		node has only one long-range connection (represented by the dashed light
		green directed lines) to make the graph less cluttered. If $Q$ were 1,
		each highway node would have 9 long-range connections.}
\end{figure}

\subsection{Results}\label{sec:khresults}

Our results depend on whether or not the structure of the highway is known to
the vertices. Due to the structured nature of the highway, we will assume that
its layout is known to all vertices (a constant amount of information), such that nodes know the location of the closest highway node to them. We
will include both results for completeness, and both have the same optimum value
and result, but our standard definition of our model will include this natural
assumption.

We split our decentralized algorithm to route from $s$ to $t$ into three
steps:

\begin{enumerate}
	\item We use local connections in $\mathcal{G}$ to route from $s$ to the
	closest highway node.

        \item We traverse the highway ($\mathcal{G}_H$) using standard Kleinberg routing towards $t$.

	\item Finally, we use the local connections in $\mathcal{G}$ to route to
	$t$.
\end{enumerate}

A straightforward proof, included
for completeness 
in \Cref{sec:khp}, produces the following result:

\begin{theorem}\label{thm:kh_routing} The expected decentralized routing time in
	a \bem{Kleinberg highway} network is $\mathcal{O}(\sqrt{k} + \log^2(n)/k +
	\log{n})$ for $1 \leq k \leq n^2$ when each node knows the positioning of
	the highway grid, and $\mathcal{O}(k + \log^2(n)/k)$ otherwise.
\end{theorem}

Reassuringly, both results are consistent with the original Kleinberg model when
$k$ is constant, with the expected routing time being $\mathcal{O}(\log^2 n)$.
Our key observation, however, is that the expected routing time reaches a global
minimum when $\Theta(\log{n}) \leq k \leq \Theta(\log^2{n})$ when the
positioning of the highway is known, or just when $k \in \Theta(\log{n})$ in
general, in which case the expected routing time becomes $\mathcal{O}(\log{n})$,
as shown visually in \Cref{fig:paramk}. 
This is a major improvement over the
original Kleinberg model.
\section{Randomized Highway}

The key difference between this model and the Kleinberg highway model is that in
this model highway nodes are distributed randomly through the entire graph
$\mathcal{G}$ instead of the unrealistic expectation that they are distributed
perfectly uniformly. As in the previous model, nodes are laid out in an $n
\times n$ grid with wrap-around, where each node is connected to its 4 directly
adjacent neighbors. Each node independently becomes a highway node with
probability $1/k$ for $1 \leq k \leq n^2/\log{n}$ such that there are an
expected $\Theta(n^2/k)$ highway nodes total w.h.p., and each highway node adds
$Q \times k$ long-distance connections to other highway nodes such there is an
expected average of $Q$ long-distance connections per node w.h.p.\footnote{This
holds for $k \in o(n^2/\log{n})$ when $k \in \Theta(n^2/\log{n})$, the density
is at most $\alpha Q$ w.h.p. for a large enough constant $\alpha$.}. As before,
each highway node only considers other highway nodes as candidates for
long-distance connections, and the probability that highway node $u$ picks
highway node $v$ as a long-distance connection is proportional to $d(u,
v)^{-2}$. An important difference, however, is that there is no clear notion of
local connections between highway nodes in this graph, which will affect the
decentralized greedy routing results. See \Cref{fig:rh}.

\begin{figure}[t]
	\centering
	\begin{tikzpicture}[
	c/.style={circle, draw=okabe1},
	h/.style={circle, draw=okabe1, fill=okabe2},
	l/.style={draw=okabe3,thick},
	lr/.style={draw=okabe4,thick,dashed,-latex},
]
	\pgfmathsetmacro{\n}{9}
	\pgfmathsetmacro{\k}{3}
	\pgfmathtruncatemacro{\npo}{\n + 1}

	\def\highwaylist{1/1,1/6,2/8,6/4,6/9,7/2,7/7,9/5,9/9}

	\foreach \x in {1,...,\n} {
		\foreach \y in {1,...,\n} {
			\pgfmathtruncatemacro{\isHighway}{0} 

			\foreach \hx/\hy in \highwaylist {
				\ifnum\hx=\x
					\ifnum\hy=\y
						\gdef\isHighway{1} 
						\breakforeach 
					\fi
				\fi
			}
			\ifnum\isHighway=1
				\node[h] (\x-\y) at (0.5*\x, 0.5*\y) {};
			\else
				\node[c] (\x-\y) at (0.5*\x, 0.5*\y) {};
			\fi
		}
	}

	\foreach \x in {0,...,\npo} {
		\node (\x-0) at (0.5*\x, 0) {};
		\node (\x-\npo) at (0.5*\x, 0.5*\npo) {};
	}

	\foreach \y in {1,...,\n} {
		\node (0-\y) at (0, 0.5*\y) {};
		\node (\npo-\y) at (0.5*\npo, 0.5*\y) {};
	}

	\foreach \x in {1,...,\n} {
		\foreach \y in {1,...,\n} {
			\pgfmathtruncatemacro{\xpo}{\x + 1}
			\pgfmathtruncatemacro{\ypo}{\y + 1}

			\draw (\x-\y) -- (\x-\ypo);
			\draw (\x-\y) -- (\xpo-\y);
		}
	}

	\foreach \i in {1,...,\n} {
		\draw (\i-0) -- (\i-1);
		\draw (0-\i) -- (1-\i);
	}

	\draw[lr] (1-1) to (7-2);
	\draw[lr] (1-6) to (9-5);
	\draw[lr] (2-8) to [bend right=15] (6-4);
	\draw[lr] (6-4) to [bend left=15] (1-6);
	\draw[lr] (6-9) to (2-8);
	\draw[lr] (7-2) to (6-4);
	\draw[lr] (7-7) to (6-9);
	\draw[lr] (9-5) to (7-2);
	\draw[lr] (9-9) to [bend left=25] (2-8);
\end{tikzpicture}
	\vspace*{-0.25cm}
	\caption{\label{fig:rh}An example of the randomized highway model with $n =
		9$, $k = 9$, and $Q = 1/9$. The solid black and curved solid blue lines
		represent local connections for the entire grid. In this model, there
		are no local connections for the highway subgraph. The value of $Q$ was
		picked such that each highway node has only one long-range connection
		(represented by the dashed light green directed lines) to make the graph
		less cluttered. If $Q$ were 1, each highway node would have 9 long-range
		connections.}
\end{figure}
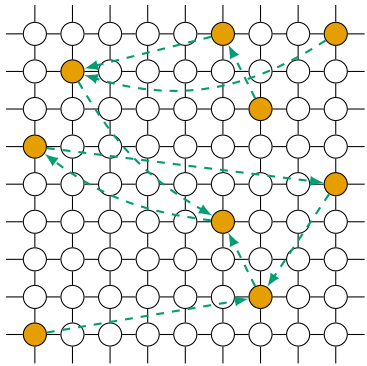

\subsection{Results}

As before, we split our decentralized routing algorithm into three steps:
reaching a highway node from $s$, traversing the highway, and reaching $t$ from
the highway. While traversing the highway, we will only take local connections
that improve our distance to $t$ by at least $4 \sqrt{k}$, for reasons that will
be clear from the proof of \Cref{lem:rh_improved_z}. We will show that the
expected time to reach a highway node from $s$ is $\mathcal{O}(k)$, the expected
time to traverse the highway is $\mathcal{O}(\log^2{n})$ for $k \in
o(\log{n})$ or $\mathcal{O}(\log{n})$ w.h.p. for $k \in \Omega(\log{n})$, and
the expected time to reach $t$ from the highway is $\mathcal{O}(k + \log{n})$
w.h.p. From these results, we will obtain:

\begin{theorem}\label{thm:rh_routing} For $k \in o\left (
	\frac{\log{n}}{\log{\log{\log{n}}}} \right )$, the expected decentralized
	greedy routing path length is $~\mathcal{O}(\log^2{n})$, while for
	$~\Theta\left ( \frac{\log{n}}{\log{\log{\log{n}}}} \right ) \leq k <
	\Theta(\log{n})$, the expected decentralized greedy routing path length is
	$\mathcal{O}(\log^2(n)/k)$, and finally for $\Theta(\log{n}) \leq k
	\leq \Theta(n)$, the expected decentralized greedy routing path length is
	$\mathcal{O}(k)$. Finally, for $k \in \Omega(n)$, the expected path length
	is $\mathcal{O}(n)$.
\end{theorem}

Note that importantly, the results of \Cref{thm:rh_routing} are worse than the
results of \Cref{thm:kh_routing} for values of $k$ between $\Theta(1)$ to
$o\left ( \frac{\log{n}}{\log{\log{\log{n}}}} \right )$, and for values of $k$
greater than $\Theta(n)$. This can be attributed to two facts, the first being
that the location of the closest highway node to $s$ is not known, and the
second being that there is no notion of local connections between the highway
nodes.
Furthermore, these bounds are not high probability bounds, since without
knowledge of the highway nodes' locations, there is a constant probability
($\approx e^{-c}$) of not reaching the highway within $c \times k$ steps, for
any constant $c$.
Allowing any non-highway node to know the location of the closest highway node
(adding a constant amount of information per node) would improve the time to get
onto the highway to at most $\mathcal{O}(\sqrt{k \log{n}})$ w.h.p. for any node,
a result derived from \Crefpart{lem:rh_nested_lattice}{lem:rh_n_lower}, but is a
variant which we do not consider further.

\subsection{Greedy Routing Sketch}\label{sec:rhsketch}

Proving the expected decentralized greedy routing path length results for the
randomized highway model in \Cref{thm:rh_routing} follows similar steps to the
proof for the Kleinberg highway model in \Cref{thm:kh_routing}.
We include a sketch below, leaving the complete proofs for the appendix in \Cref{subsec:rh_proofs}.

We start by proving a lower bound on the probability that a
long-range connection exists between two arbitrary highway nodes. In order to do
this, we need to find a high probability upper bound on the normalization
constant $z$ for any arbitrary highway node.

\begin{lemma}\label{lem:rh_norm_loose} The normalization constant $z$ for any
	arbitrary highway node is at most $25 \log{\log{\log{n}}} + \frac{41}{9}
	\frac{\log{n}}{k \log{\log{n}}} + 26\frac{\log{n}}{k}$ for $n > 5$ w.h.p.
	(for at most $\mathcal{O}(\log^2{n})$ invocations).
\end{lemma}

This result gives us a normalization constant that is in
$\mathcal{O}(\log(n)/k)$ for $k \in o\left
(\frac{\log{n}}{\log{\log{\log{n}}}}\right )$, and in
$\mathcal{O}(\log{\log{\log{n}}})$ for $k \in \Omega\left
(\frac{\log{n}}{\log{\log{\log{n}}}}\right )$. Note that this bound is worse for
large values of $k$ than the bound we obtained for the Kleinberg highway model
in \Cref{lem:kh_norm}. We can, however, improve this bound, but without the same
high probability guarantees:

\begin{lemma}\label{lem:rh_norm_tight} The normalization constant $z$ for any
	arbitrary highway node is at most $10 + 37 \frac{\log{n}}{k}$ for $n > 2$
	with probability at least $1/2$. From now on, we will refer to this tighter
	bound as $z'$.
\end{lemma}

This improved bound gives us a normalization constant that is in
$\mathcal{O}(\log(n)/k)$ for $k \in o(\log{n})$ in $\mathcal{O}(1)$ for $k \in
\Omega(\log{n})$, a result in line with the Kleinberg highway model. We want to
be able to use this improved bound when calculating the probability of halving
the distance to the destination.

\begin{lemma}\label{lem:rh_improved_z} Using the improved normalization constant
	bound $z'$ incurs at most a constant factor to the probability of halving
	the distance to the destination while routing w.h.p.
\end{lemma}

Now we can use these improved normalization constant bounds to find the
probability of halving our distance. Suppose we are in phase $j$ where $\log(c(k
+ \log{n})) \leq j \leq \log{n}$ (for some constant $c$ we will discuss later),
and the current message holder $u$ is a highway node. Let us find the
probability that we have a long-range contact that is in a better phase. First,
we find the number of highway nodes in a better phase than us, i.e., within the
ball of radius $2^j$ around $t$ ($B_{2^j}(t)$).

\begin{lemma}\label{lem:rh_ballj} There are at least $2^{2j - 2} / k$ highway
	nodes in a ball of radius $2^j$ for $\log(c(k + \log{n})) \leq j \leq
	\log{n}$ with high probability (with probability at least $1 -
	n^{-0.18c^2}$).
\end{lemma}

Each of these nodes has lattice distance less than $2^{j + 2}$, allowing us to
bound the probability of them being a specific long-range contact of $u$. Then,
we can obtain an identical result (in asymptotic notation) to the result in
\Cref{lem:kh_halving}:

\begin{lemma}\label{lem:rh_halving} In the randomized highway model, the
    probability that a node $u$ has a long-range connection to a node $v$ that
    halves its distance to the destination is proportional to at least
    $k/\log{n}$ for $k \in \mathcal{O}(\log{n})$ and is constant for $k \in
    \Omega(\log{n})$.
\end{lemma}

Once we reach phase $j = \log(c(k + \log{n}))$, we are at distance
$\mathcal{O}(k + \log{n})$ from the destination, reaching it in $\mathcal{O}(k + \log{n})$ local hops. As stated up until now, we would be
able to perform greedy routing with results equivalent to those of
\Cref{thm:kh_routing} assuming no knowledge about the positioning of the highway
nodes ($\mathcal{O}(k + \log^2(n)/k)$ routing). 
However, we have not yet
addressed the elephant in the room: the fact that there is no notion of local
contacts between highway nodes. In simple terms, while routing, if there are no
long-range contacts that improve your distance, you must leave the highway. And
when you leave the highway, it may take a while to get back onto it. We will
show that this is not a problem for large values of $k$, i.e. values of $k \in
\Omega(\log{n})$, but for smaller values of $k$ the bound will be worse than
before, becoming $\mathcal{O}(\log^2{n})$ expected routing instead of
$\mathcal{O}(\log^2(n)/k)$ (note that we do not prove that the bound is tight).
In \Cref{sec:rhvariant} we propose a variant which trivially achieves the improved $\mathcal{O}(\log^2(n)/k)$ expected routing for small values of $k$. We consider this 
variant slightly less elegant, and since it maintains the same optimal results, we do not consider it further.
\section{Windowed Neighborhood Preferential Attachment}\label{sec:wnpa}

Our previous models have a binary distinction between highway nodes and normal
nodes, represented by a fixed value of $k$. We now describe a new model with a
continuous transition, where each node picks its own value of $k$, such that the
distribution of the values of $k$, and consequently the degree distribution,
exhibits a power law. Each node independently picks their probability $k$ from a
distribution $\Pr(k) \propto 1/k^{2+\epsilon}$ for $\epsilon > 0$. Each node $u$
then adds $\epsilon Q \times k$ long-range connections, but only to nodes within
a given range, or ``window'', of popularity. Specifically, let the window of
popularity for a given node $u$ with popularity $k_u$ be popularities in the
range $[k_u/A, A k_u]$.

\subsection{Results}

While at first glance this model may seem irreconcilable from the previous models, consider referring to all nodes with popularity $\log{n} \leq k \leq A \log{n}$ as the ``highway''. We expect to have $\mathcal{O}(1 / \log^{1 + \epsilon}{n})$ highway nodes. Ignoring all long-range connections that do not connect two highway nodes, we find an instance of the randomized highway model embedded within the windowed NPA model, albeit with a small (but nevertheless constant) value of $Q$. With these key observations, we are able to prove:

\begin{theorem}\label{thm:bak_routing}
	The windowed NPA model has a decentralized greedy algorithm that routes in
	expected $\mathcal{O}(\log^{1 + \epsilon}(n))$ hops.
\end{theorem}

The complete proof for this theorem can be found in \Cref{sec:wnpaproof}. Furthermore, experimental results confirming that this model greedily routes significantly better than Kleinberg's can be found in \Cref{sec:experiments}.

\subsection{Efficient Construction}

The neighborhood preferential attachment model of Goodrich and Ozel
\cite{goodrichModelingSmallworldPhenomenon2022} takes $\mathcal{O}(|V|^2)$ time to construct and
there is no more efficient construction currently known. The windowed NPA model
can similarly be constructed sequentially in $\mathcal{O}(|V|^2)$ time. However,
due to how each node picks their connections independently, this model is
embarrassingly parallel, and can be constructed in $\mathcal{O}(|V|)$
time with $|V|$ processors, without any communication between processors.



\section{Future Work}

It would be interesting to be able to prove whether our bounds are tight for our
models. Specifically, whether the bounds for the randomized highway model can be
improved to be more in line with the Kleinberg highway results%
.
While the diameter of models with constant degree is
at least 
$\Omega(\log{n})$, there is no such lower bound when dealing with constant \emph{average} degree. It would be interesting to either bridge the gap or show that a true gap exists between the lower bound on the diameter of our networks, $\Omega(\log{n}/\log{\log{n}})$, and the upper bound on greedy routing, $\mathcal{O}(\log{n})$.
Also, it would be
interesting to prove whether it is possible to achieve a greedy routing time of
$\log{n} + \mathbf{\sqrt{k}}$ for larger values of $k$ if each node knows the
location of the nearest highway node
(a constant amount of
additional
information). This result would improve the expected running time of the
windowed NPA model to just $\mathcal{O}(\log{n})$ for $0 < \epsilon \leq 1$.
Finally, our analysis for the randomized highway model depends on the network
having a mostly even spread of nodes. Experimentally, both our model and the original NPA model
perform worse on Alaska, a highly unevenly spread out state. It would be
interesting to generalize our results if some form of density condition is met.

\nocite{*}
\bibliographystyle{splncs04}
\bibliography{refs}

\begin{thebibliography}{10}
\providecommand{\url}[1]{\texttt{#1}}
\providecommand{\urlprefix}{URL }
\providecommand{\doi}[1]{https://doi.org/#1}

\bibitem{barabasiEmergenceScalingRandom1999}
Barab{\'a}si, A.L., Albert, R.: Emergence of scaling in random networks. Science  \textbf{286}(5439),  509--512 (1999). \doi{10.1126/science.286.5439.509}

\bibitem{DBLP:conf/icalp/BergerBCDK04}
Berger, N., Borgs, C., Chayes, J.T., D'Souza, R.M., Kleinberg, R.D.: Competition-induced preferential attachment. In: D{\'{\i}}az, J., Karhum{\"{a}}ki, J., Lepist{\"{o}}, A., Sannella, D. (eds.) Automata, Languages and Programming: 31st International Colloquium, {ICALP} 2004, Turku, Finland, July 12-16, 2004. Proceedings. Lecture Notes in Computer Science, vol.~3142, pp. 208--221. Springer (2004). \doi{10.1007/978-3-540-27836-8\_20}, \url{https://doi.org/10.1007/978-3-540-27836-8\_20}

\bibitem{BR05}
Bollob{\'a}s, B., Riordan, O.M.: Mathematical results on scale-free random graphs. In: Bornholdt, S., Schuster, H.G. (eds.) Handbook of Graphs and Networks: From the Genome to the {Internet}, chap.~1, pp. 1--34. Wiley (2002). \doi{https://doi.org/10.1002/3527602755.ch1}

\bibitem{DBLP:conf/stoc/BorgsCDR07}
Borgs, C., Chayes, J.T., Daskalakis, C., Roch, S.: First to market is not everything: an analysis of preferential attachment with fitness. In: Johnson, D.S., Feige, U. (eds.) Proceedings of the 39th Annual {ACM} Symposium on Theory of Computing, San Diego, California, USA, June 11-13, 2007. pp. 135--144. {ACM} (2007). \doi{10.1145/1250790.1250812}, \url{https://doi.org/10.1145/1250790.1250812}

\bibitem{doddsExperimentalStudySearch2003}
Dodds, P.S., Muhamad, R., Watts, D.J.: An experimental study of search in global social networks. Science  \textbf{301}(5634),  827--829 (2003). \doi{10.1126/science.1081058}, \url{https://www.science.org/doi/abs/10.1126/science.1081058}

\bibitem{dommersDiametersPreferentialAttachment2010}
Dommers, S., van~der Hofstad, R., Hooghiemstra, G.: Diameters in preferential attachment models. Journal of Statistical Physics  \textbf{139}(1),  72--107 (2010). \doi{10.1007/s10955-010-9921-z}

\bibitem{DBLP:journals/im/FlaxmanFV07}
Flaxman, A.D., Frieze, A.M., Vera, J.: A geometric preferential attachment model of networks. Internet Math.  \textbf{3}(2),  187--205 (2007). \doi{10.1080/15427951.2006.10129124}

\bibitem{goodrichModelingSmallworldPhenomenon2022}
Goodrich, M.T., Ozel, E.: Modeling the small-world phenomenon with road networks. In: Renz, M., Sarwat, M. (eds.) Proceedings of the 30th International Conference on Advances in Geographic Information Systems, {SIGSPATIAL} 2022, Seattle, Washington, November 1-4, 2022. pp. 46:1--46:10. {ACM} (2022). \doi{10.1145/3557915.3560981}, \url{https://doi.org/10.1145/3557915.3560981}

\bibitem{kleinbergSmallworldPhenomenonAlgorithmic2000a}
Kleinberg, J.M.: The small-world phenomenon: an algorithmic perspective. In: Yao, F.F., Luks, E.M. (eds.) Proceedings of the Thirty-Second Annual {ACM} Symposium on Theory of Computing, May 21-23, 2000, Portland, OR, {USA}. pp. 163--170. {ACM} (2000). \doi{10.1145/335305.335325}

\bibitem{DBLP:conf/esa/KumarLT06}
Kumar, R., Liben{-}Nowell, D., Tomkins, A.: Navigating low-dimensional and hierarchical population networks. In: Azar, Y., Erlebach, T. (eds.) Algorithms - {ESA} 2006, 14th Annual European Symposium, Zurich, Switzerland, September 11-13, 2006, Proceedings. Lecture Notes in Computer Science, vol.~4168, pp. 480--491. Springer (2006). \doi{10.1007/11841036\_44}

\bibitem{DBLP:journals/pnas/Liben-NowellN0R05}
Liben{-}Nowell, D., Novak, J., Kumar, R., Raghavan, P., Tomkins, A.: Geographic routing in social networks. Proc. Natl. Acad. Sci. {USA}  \textbf{102}(33),  11623--11628 (2005). \doi{10.1073/pnas.0503018102}

\bibitem{martelAnalyzingKleinbergOther2004a}
Martel, C.U., Nguyen, V.: Analyzing {Kleinberg's} (and other) small-world models. In: Chaudhuri, S., Kutten, S. (eds.) Proceedings of the Twenty-Third Annual {ACM} Symposium on Principles of Distributed Computing, {PODC} 2004, St. John's, Newfoundland, Canada, July 25-28, 2004. pp. 179--188. {ACM} (2004). \doi{10.1145/1011767.1011794}

\bibitem{milgramstanleySmallWorldProblem1967}
Milgram, S.: The small world problem. Psychology Today  \textbf{1}(1),  61--67 (1967)

\bibitem{mitzenmacher2004brief}
Mitzenmacher, M.: A brief history of generative models for power law and lognormal distributions. Internet Mathematics  \textbf{1}(2),  226--251 (2004)

\bibitem{DBLP:conf/podc/Slivkins05}
Slivkins, A.: Distance estimation and object location via rings of neighbors. In: Aguilera, M.K., Aspnes, J. (eds.) Proceedings of the Twenty-Fourth Annual {ACM} Symposium on Principles of Distributed Computing, {PODC} 2005, Las Vegas, NV, USA, July 17-20, 2005. pp. 41--50. {ACM} (2005). \doi{10.1145/1073814.1073823}

\bibitem{traversExperimentalStudySmall1969}
Travers, J., Milgram, S.: An experimental study of the small world problem. Sociometry  \textbf{32}(4),  425--443 (1969)

\end{thebibliography}

\clearpage

\section{Appendix}

\subsection{Experimental Analysis}\label{sec:experiments}

Goodrich and Ozel's paper on the neighborhood preferential model
\cite{goodrichModelingSmallworldPhenomenon2022} was able to show that a hybrid
model combining elements from Kleinberg's model with preferential attachment is
able to outperform both individual models for decentralized greedy routing on
road networks by showing many experimental results. In the previous sections,
we provided some theoretical justification for their results, by proving
asymptotically better greedy routing times for a similar model. In this section,
we complete our comparisons by reproducing their key experimental results with
our new model. Our experimental framework is nearly identical to theirs, except
that we implement directed versions of each algorithm, i.e. where each
long-range connection is directed (local connections are by definition always
undirected). This allows us to run experiments much more efficiently---we sample
between 30,000 to 200,000 source/target pairs for each data point, as compared
to their 1,000 pairs---but results in all algorithms having a worse performance.
For our experiments we picked $\epsilon = 0.5$ and $A = 1.01$. It is possible
that other parameters would yield better results.

\subsubsection{Key Results}

\begin{figure}[b!]
    \centering
    \includegraphics[width=\columnwidth]{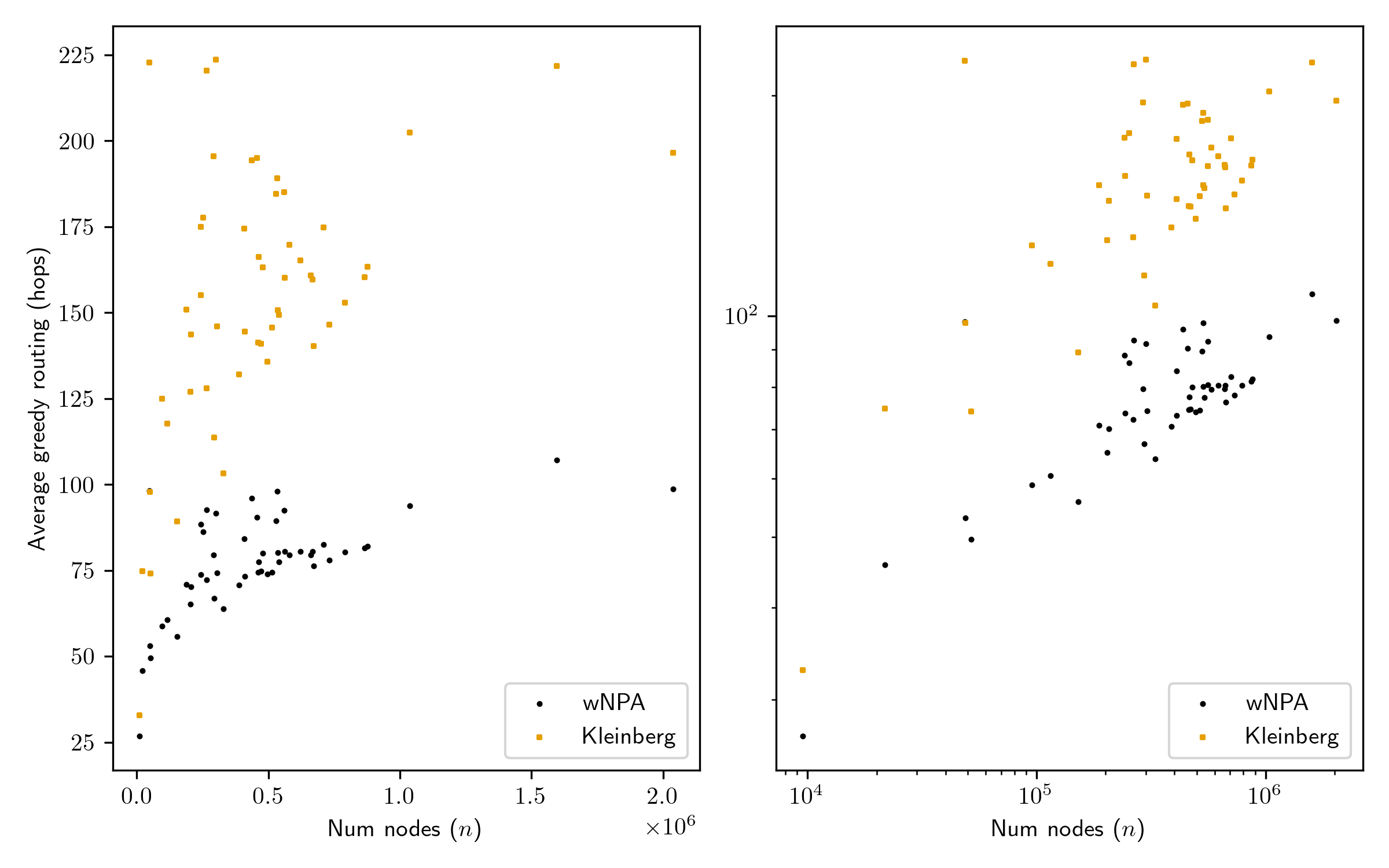}
    \caption{ Comparison of greedy routing times for Kleinberg's model and the
        windowed NPA model when $Q = 1, \epsilon = 0.5, A = 1.01$. The right plot is in
        log scale. }
    \label{fig:kleinberg_comp}
\end{figure}

Our main key result is that our windowed NPA model outperforms Kleinberg's model for road
networks by a factor of 2, as shown in \Cref{fig:kleinberg_comp}. This result is
directly in line with Goodrich and Ozel's experimental results with their
similar model \cite{goodrichModelingSmallworldPhenomenon2022}. It is worth
mentioning that our directed version of the model is worse than the undirected
version from Goodrich and Ozel's paper by roughly a factor of 2.

Similarly, we show that by increasing the degree density to 32 we can achieve a
result of less than 20 degrees of separation, which again is roughly twice the
results from Goodrich and Ozel's paper (see
\Cref{fig:greedy_routing_average_degree}), which we attribute primarily to the
directed implementation of the models for our experiments.

\begin{figure}[t!]
    \centering
    \includegraphics[width=\columnwidth]{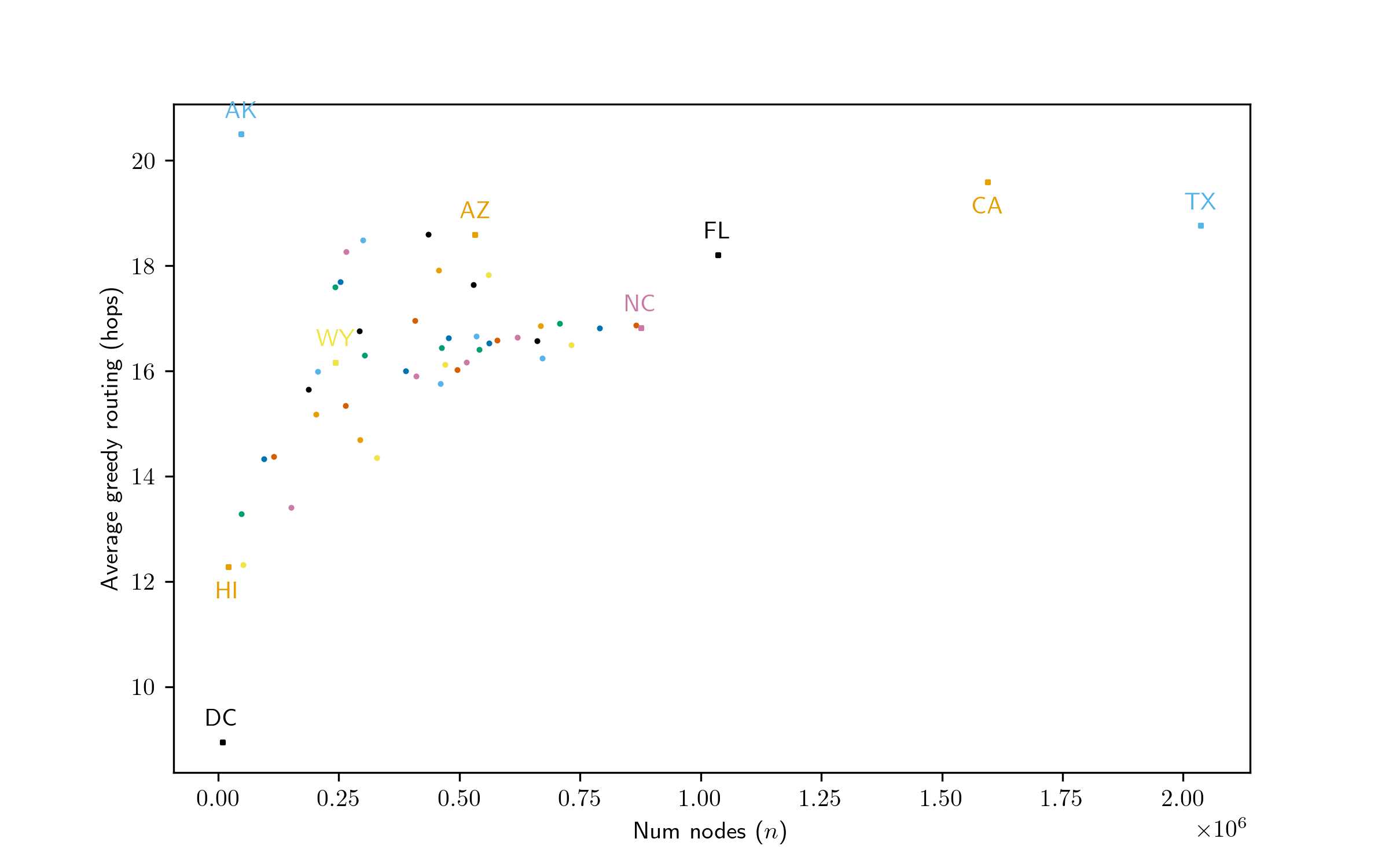}
    \caption{ The greedy routing times for the windowed NPA model on the 50 US states
        when $Q = 32$, $\epsilon = 0.5$, and $A = 1.01$. }
    \label{fig:greedy_routing_average_degree}
\end{figure}

\subsection{Kleinberg Highway Proofs}\label{sec:khp}

In this section, we prove \Cref{thm:kh_routing} by proving upper bounds on each of the three steps of the greedy routing algorithm: routing from $s$ to the highway using local connections, within the highway towards $t$ using standard Kleinberg routing, and finally from the highway to $t$ again using local connections.

\begin{lemma}\label{lem:khrouting_to_highway} It is possible to route from any
	node $s \in \mathcal{G}$ to a highway node $h \in \mathcal{G}_H$ in at most
	$\sqrt{k}$ hops, if the location of $h$ is known, or in at most $k - 1$
	hops, if the location of $h$ is not known.
\end{lemma}

\begin{proof}
	Without loss of generality, let's assume highway nodes are located wherever
    $\bmod(x, \sqrt{k}) = 0$ and $\bmod(y, \sqrt{k}) = 0$. Then, the maximum
    distance in the $x$ dimension to a highway node is $\delta_x = \min(\bmod(x,
    \sqrt{k}), \sqrt{k} - \bmod(x, \sqrt{k})) = \left \lfloor \frac{\sqrt{k}}{2}
    \right \rfloor$, and an equivalent result holds for $\delta_y$. Therefore,
    the maximum lattice distance to a highway node is the sum of both, or at
    most $2 \left \lfloor \frac{\sqrt{k}}{2} \right \rfloor \leq \sqrt{k}$. If
    the location of $h$ is known, then we can route to it directly taking a
    number of hops equal to the lattice distance to $h$. If the location of $h$
    is not known, we can visit every node in a $\sqrt{k} \times \sqrt{k}$
    square, guaranteeing that we will encounter a highway node $h$, in $k - 1$
    hops.
\end{proof}

After we reach the highway subgraph $\mathcal{G}_H$, we can use the standard
Kleinberg routing algorithm towards $t$. As in Kleinberg's original analysis, we
first prove a lower bound on the probability that a long-range connection exists
between two arbitrary highway nodes.

\begin{lemma}\label{lem:kh_norm} The normalization constant $z$ for
	$\mathcal{G}_H$ is upper bounded by $z \leq 4 \ln(6 n_H) \leq 4 \ln(6 n)$.
	As such, the probability of any two highway nodes $u$ and $v$ being
	connected is at least $[4 \ln(6 n) d_H(u, v)^2]^{-1}$, where $d_H(u, v)$ is
	the lattice distance between $u$ and $v$ in $\mathcal{G}_H$.
\end{lemma}

\begin{proof}
	This result follows directly from Kleinberg's original analysis on the
	highway subgraph $\mathcal{G}_H$.
\end{proof}

In Kleinberg's analysis, the probability that a node $u$ has a long-range
connection to a node $v$ that halves its distance to the destination is
proportional to $[\log{n}]^{-1}$, when a node has a constant number of
long-range connections $Q$. In our case, each highway node has $Q \times k$
long-range connections, where $k$ does not need to be constant. This gives us
improved distance-halving probabilities:

\begin{lemma}\label{lem:kh_halving} In the Kleinberg highway model, the
	probability that a node $u$ has a long-range connection to a node $v$ that
	halves its distance to the destination is proportional to at most
	$k/\log{n}$ for $k \in \mathcal{O}(\log{n})$ and is constant for $k \in
	\Omega(\log{n})$.
\end{lemma}

\begin{proof}
	Following Kleinberg's analysis, the probability that a single long-range
	connection from $u$ halves its distance to the destination is still
	proportional to $[\log{n}]^{-1}$. Therefore, the probability that a single
	long-range connection does \textit{not} halve its distance to the
	destination is proportional to $1 - [\log{n}]^{-1}$. The probability that
	all $Qk$ long-range connections do not halve the distance is therefore
	proportional to $\left (1 - [\log{n}]^{-1} \right )^{Qk} = \left [ \left (1
	- [\log{n}]^{-1} \right )^{\log{n}} \right ]^{\frac{Qk}{\log{n}}} \leq
	e^{-\frac{Qk}{\log{n}}}$. Finally, the probability that any one of the $Qk$
	succeed in halving the distance is therefore proportional to $1 -
	e^{-\frac{Qk}{\log{n}}}$. When $k \in \omega(\log{n})$, the exponential term
	tends towards zero, and the probability tends towards one. For smaller
	values of $k$, a Taylor expansion of $e^{-\frac{Qk}{\log{n}}}$ shows that
	this probability is proportional to at least $1 - \left [ 1 -
	\frac{Qk}{\log{n}} + \mathcal{O}\left ( \left [ \frac{Qk}{\log{n}} \right
	]^2 \right ) \right ] = \frac{Qk}{\log{n}} - \mathcal{O}\left ( \left [
	\frac{Qk}{\log{n}} \right ]^2 \right )$. When $k \in o(\log{n})$, the lower
	order terms become asymptotically negligible, and we are left with a
	probability proportional to $\frac{Qk}{\log{n}} = \mathcal{O}(k/\log{n})$.
	When $k = \Theta(\log{n})$, we are left with a constant dependent on $Q$.
\end{proof}

Importantly, this result reproduces Kleinberg's original result when $k$ is
constant, since we are left with a probability proportional to $1/\log{n}$.
Finally, we can prove the main result of this section:

\begin{proof}[of \Cref{thm:kh_routing}]
	It is possible to describe the greedy routing path in terms of at most
	$\log{n}$ phases, where a node $u$ in phase $j$ if it is at a lattice
	distance between $2^j$ and $2^{j+1}$ from the destination $t$. It is easy to
	see that halving the distance to the destination results in reducing what
	phase a node is in by one. The expected amount of hops spent in each phase
	is therefore $1 / \Pr(\text{distance halving}) = \mathcal{O}(\log(n)/k)$.
	Note that importantly, when no long-range connections halve the distance, we
	take local connections on the \textit{highway graph} towards $t$, as in the
	original Kleinberg model. Since there are at most $\log{n}$ phases, we
	expect to spend at most $\mathcal{O}(\log{n}(\log(n)/k + 1))$ hops on the
	highway\footnote{Some minor details regarding the final $\log{\log{n}}$
	phases are omitted for brevity.}. Finally, the final highway node is known
	to be at most $\sqrt{k}$ hops away from the destination $t$. The theorem
	follows from these results along with the results from
	\Cref{lem:khrouting_to_highway}.
\end{proof}

\subsection{Randomized Highway Proofs}\label{subsec:rh_proofs}

We now present proofs of theorems and lemmas discussed in \Cref{sec:rhsketch}.

\subsubsection{The Nested Lattice Construction}

For our proofs, similarly to the Kleinberg highway model, we will conceptually
subdivide the highway into a lattice of balls of various sizes (see
\Cref{fig:nestedlattice} for an example nested lattice structure), and show
upper and lower bounds on the number of highway nodes within each ball with
varying degrees of probability bounds. Specifically we will prove:

\begin{lemma}\label{lem:rh_nested_lattice} Results from the nested lattice
	structure:
	\begin{enumerate}
		\item All balls of radius $3 \sqrt{k \log{n}}$, centered around any of
		the $n^2$ nodes, contain at least $9 \log{n}$ highway nodes with high
		probability in $n$. \label{lem:rh_n_lower}

		\item All balls of radius $3 \sqrt{k \log{n}}$, centered around any of
		the $n^2$ nodes, contain fewer than $41 \log{n}$ highway nodes with high
		probability in $n$. \label{lem:rh_n_upper}

		\item $\mathcal{O}(\log^2{n})$ balls of radius $3 \sqrt{k
		\log{\log{n}}}$, centered \\ around any $\mathcal{O}(\log^2{n})$ nodes,
		contain fewer than $41 \log{\log{n}}$ highway nodes with high
		probability in $\log{n}$. \label{lem:rh_logn_upper}

		\item Any arbitrary ball of radius $2 \sqrt{k}$ has at most 18 highway
		nodes with probability at least $1/2$. This result is \textit{not} a
		high probability bound, and is only independent for balls centered
		around nodes with lattice distance greater than $4 \sqrt{k}$ between
		them. \label{lem:rh_const_upper}
	\end{enumerate}
\end{lemma}

\begin{figure}
	\centering
	\begin{tikzpicture}[
	d/.style={circle, inner sep=0pt, minimum size=1.5pt, fill=okabe1},
	ctr/.style={circle, inner sep=0pt, minimum size=3pt, fill=okabe2},
	b/.style={draw=okabe3, thick},
	c/.style={draw=okabe4, fill=okabe4, opacity=0.5},
	n/.style={draw=none, pattern=crosshatch, pattern color=okabe5, opacity=0.5},
]

	\pgfmathsetmacro{\w}{29} 
	\pgfmathsetmacro{\h}{19} 
	\pgfmathsetmacro{\s}{0.25} 

	\pgfmathtruncatemacro{\cx}{\w/2+1} 
	\pgfmathtruncatemacro{\cy}{\h/2+1} 

	\pgfmathsetmacro{\br}{3} 
	\pgfmathtruncatemacro{\bd}{2*\br} 

	\foreach \x in {1,...,\w} {
		\foreach \y in {1,...,\h} {
			\coordinate (\x-\y) at (\s*\x, \s*\y);
		}
	}

	\pgfmathtruncatemacro{\nxp}{\cx + \br + \bd}
	\pgfmathtruncatemacro{\nxm}{\cx - \br - \bd}
	\pgfmathtruncatemacro{\nyp}{\cy + \br + \bd}
	\pgfmathtruncatemacro{\nym}{\cy - \br - \bd}

	\pgfmathtruncatemacro{\cxpbr}{\cx + \br}
	\pgfmathtruncatemacro{\cxmbr}{\cx - \br}
	\pgfmathtruncatemacro{\cypbr}{\cy + \br}
	\pgfmathtruncatemacro{\cymbr}{\cy - \br}

	\pgfmathtruncatemacro{\l}{0}
	\pgfmathtruncatemacro{\rx}{\w+1}
	\pgfmathtruncatemacro{\ry}{\h+1}

	\tikzset{
		clip even odd rule/.code={\pgfseteorule}, 
		invclip/.style={
			clip,insert path=
				[clip even odd rule]{
					[reset cm](-\maxdimen,-\maxdimen)rectangle(\maxdimen,\maxdimen)
				}
		}
	}

	\begin{scope}
		\begin{pgfinterruptboundingbox}
			\clip[invclip] (\cx-\cypbr) -- (\cxpbr-\cy) -- (\cx-\cymbr) -- (\cxmbr-\cy) -- cycle;
		\end{pgfinterruptboundingbox}

		\draw[n] (\cx-\nyp) -- (\nxp-\cy) -- (\cx-\nym) -- (\nxm-\cy) -- cycle;
	\end{scope}

	\draw[c] (\cx-\cypbr) -- (\cxpbr-\cy) -- (\cx-\cymbr) -- (\cxmbr-\cy) -- cycle;

	\pgfmathtruncatemacro{\topx}{mod(\cx + \cy + 3, \bd)}
	\pgfmathtruncatemacro{\topxs}{\topx + \bd}

	\pgfmathsetmacro{\min}{0}
	\pgfmathsetmacro{\maxx}{\w + 1}
	\pgfmathsetmacro{\maxy}{\h + 1}
	\pgfmathsetmacro{\max}{\maxx + \maxy}

	\clip (\s*\l, \s*\l) rectangle (\s*\rx, \s*\ry);

	\pgfmathtruncatemacro{\botx}{mod(\cx + \cy + 3, \bd)}
	\pgfmathtruncatemacro{\botxs}{\botx + \bd}

	\foreach \x in {\botx,\botxs,...,\max} {
		\pgfmathsetmacro{\y}{\maxy - \x}

		\draw[b] (\s*\x, \s*\maxy) -- (\s*\min, \s*\y);
	}

	\foreach \x in {\topx,\topxs,...,\max} {
		\draw[b] (\s*\x, \s*\min) -- (\s*\min, \s*\x);
	}

	\foreach \x in {1,...,\w} {
		\foreach \y in {1,...,\h} {
			\ifthenelse{\x = \cx \AND \y = \cy}{
				\node[ctr] at (\x-\y) {};
			}{
				\node[d] at (\x-\y) {};
			}
		}
	}
	
\end{tikzpicture}
	\caption{\label{fig:nestedlattice} The nested lattice construction showing
		balls of radius 3, centered around an orange node. The central ball is
		depicted in solid light green, while the 8 adjacent balls are shown in
		dashed yellow.}
\end{figure}
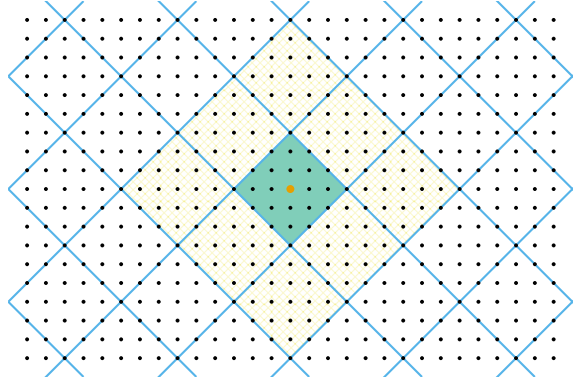

\begin{proof}
	Consider balls of radius $a \sqrt{k \log{n}}$ for some constant $a$. There
	are at least $2 a^2 k \log{n}$-many nodes within each ball of radius $a
	\sqrt{k \log{n}}$. The probability that any node is a highway node is $1/k$,
	so the expected number of highway nodes within each ball is $\mu \ge 2 a^2
	\log{n}$. We can lower bound the number of highway nodes within each ball by
	using a Chernoff bound. Letting $X$ be the number of highway nodes within
	each ball, we have:
	\begin{equation*}
		\Pr(X \leq (1 - \delta)\mu) \leq e^{-\frac{\delta^2 \mu}{2}} = e^{-a^2
		\delta^2 \log{n}} = n^{-\frac{a^2 \delta^2}{\ln{2}}}
	\end{equation*}

	By union bound, the probability this fails for a ball centered at any of the
	$n^2$ vertices is at most $n^{2 - \frac{a^2 \delta^2}{\ln{2}}}$. Setting
	$\delta = 1/2$ and $a = 3$, we obtain that all balls with radius $3 \sqrt{k
	\log{n}}$ have at least $9 \log{n}$ highway nodes with probability at least
	$1 - n^{-1.24}$, which is w.h.p. For an upper bound, we first note that
	there are fewer than $3 a^2 k \log{n}$-many nodes within each ball of radius
	$a \sqrt{k \log{n}}$ for radii of at least 3. Using another Chernoff bound:
	\begin{equation*}
		\Pr(X \geq (1 + \delta)\mu) \leq e^{-\frac{\delta^2 \mu}{2 + \delta}} =
		e^{-\frac{3 a^2 \delta^2 \log{n}}{2 + \delta}} = n^{-\frac{3 a^2
		\delta^2}{\ln{2}(2 + \delta)}}
	\end{equation*}

	By setting $\delta = 1/2$ and $a = 3$, we obtain that all balls with radius
	$3 \sqrt{k \log{n}}$ have fewer than $41 \log{n}$ highway nodes w.h.p. (with
	probability at least $1 - n^{-1.89}$). We can obtain similar bounds for
	smaller balls, although with worse probabilities. For example, for balls of
	radius $a \sqrt{k \log{\log{n}}}$, we expect $\mu < 3 a^2 \log{\log{n}}$
	highway nodes for radii of at least 3. Using another Chernoff bound with
	$\delta = 1/2$ and $a = 3$, we obtain that any given ball with radius $3
	\sqrt{k \log{\log{n}}}$ has more than $41 \log{\log{n}}$ highway nodes with
	probability less than $\log^{-3.89}{n}$. Assuming we will only invoke this
	bound at most $\mathcal{O}(\log^2{n})$ times, the probability that any of
	the invocations fail is negligible (at most $\mathcal{O}(\log^{-1.89}{n})$).
	Finally, we consider balls of radius only $2 \sqrt{k}$, which have at most
	$18$ highway nodes with probability at least $1/2$.
\end{proof}


\subsubsection{Finding the Normalization Constant}

The probability that highway node $u$ picks highway node $v$ as a long-range
connection is $d(u, v)^{-2}/ \left [\sum_{w \neq u}{d(u, w)^{-2}} \right ]$,
where each $w$ in the summation is a highway node. In order to lower bound this
probability, we must upper bound the denominator, known as the
\bem{normalization constant} $z$.

\begin{proof}[of \Cref{lem:rh_norm_loose}]
	Let's consider a lattice of balls centered around an arbitrary highway node
	$u$. Let's define a notion of ``ball distance'' $b$ to measure the distance
	between two balls in this ball lattice. Let $\mathcal{B}_b(u)$ be the set of
	all balls at ball distance $b$ from a ball centered at $u$. There is 1 ball
	at ball distance 0 ($|\mathcal{B}_0(u)| = 1$), 8 balls at ball distance 1,
	and in general at most $8b$ balls at distance $b$ for $b > 0$ (see
	\Cref{fig:nestedlattice}). The minimum distance between $u$ to a node in
	another ball at distance $b$ is $2b - 1$ times the ball radius for $b > 0$.
	Let's consider a lattice of balls with radius $3 \sqrt{k \log{n}}$. From
	\Crefpart{lem:rh_nested_lattice}{lem:rh_n_upper} we know that there are at
	most $41 \log{n}$ highway nodes within this ball w.h.p. Let's also find the
	normalization constant in two parts, first due to highway nodes in $b > 0$
	($z_{>0}$), and then due to highway nodes within the same ball ($z_0$).

	Note that any two balls are separated by ball distance at most $2n/$twice
	the ball radius, or $\frac{n}{3 \sqrt{k \log{n}}}$.
	\begin{align*}
		z_{>0} &\leq \sum_{b = 1}^{\frac{n}{3 \sqrt{k \log{n}}}}{
			\frac{(\text{max \# highway nodes in } \mathcal{B}_b(u))}{(\text{min distance to node in }\mathcal{B}_b(u))^2}} \\
		&\leq \sum_{b = 1}^{\frac{n}{3 \sqrt{k \log{n}}}}{
			\frac{41 \log{n} \times 8b}{(2b - 1)^2 \times 9 k \log{n}}}
		< \frac{37}{k} \sum_{b = 1}^{\frac{n}{3 \sqrt{k \log{n}}}}{
			\frac{b}{(2b - 1)^2}} \\
		&\leq \frac{37}{k} \sum_{b = 1}^{\frac{n}{3 \sqrt{k \log{n}}}}{
			\frac{1}{b}}
		= \frac{37}{k} \mathcal{H}\left (\frac{n}{3 \sqrt{k \log{n}}}\right ) \\
		&\leq \frac{37}{k} \mathcal{H}\left (\frac{n}{3 \sqrt{\log{n}}}\right )
		< 26\frac{\log{n}}{k} \text{ for } n > 2
	\end{align*}

	Now that we showed the contribution of highway nodes in different balls from
	$u$, let's bound the contribution due to highway nodes within the same ball.
	We are only interested in the normalization constant for nodes that we visit
	along the highway, which we will show is at most $\mathcal{O}(\log^2{n})$
	nodes. Knowing this, we can use the improved bound for balls of radius $3
	\sqrt{k \log{\log{n}}}$, which from
	\Crefpart{lem:rh_nested_lattice}{lem:rh_logn_upper} we know contain fewer
	than $41 \log{\log{n}}$ highway nodes w.h.p. Let's consider the worst case
	where they are all bunched up around $u$. Let's denote their contribution
	$z_{0, \text{inner}}$.
	\begin{align*}
		z_{0, \text{inner}} &\leq \sum_{j = 1}^{\lceil \sqrt{41 \log{\log{n}}} \rceil}{
			\frac{4j}{j^2}}
		< 4 \mathcal{H}\left ( \sqrt{41 \log{\log{n}}} + 1 \right ) \\
		&< 25 \log{\log{\log{n}}} \text{ for } n > 5
	\end{align*}

	Recall that we can still have up to $41 \log{n}$ highway nodes in in the
	same (large) ball as $u$. Let's assume they are all as close as possible,
	meaning that they are all at the edge of the inner ball. Let's denote their
	contribution $z_{0, \text{outer}}$.
	\begin{equation*}
		z_{0, \text{outer}} < \frac{41 \log{n}}{(3 \sqrt{k \log{\log{n}}})^2}
		= \frac{41}{9} \frac{\log{n}}{k \log{\log{n}}}
	\end{equation*}

	Combining these results, we obtain:
	\begin{equation*}
		z 
		< 25 \log{\log{\log{n}}} + \frac{41}{9} \frac{\log{n}}{k \log{\log{n}}} + 26\frac{\log{n}}{k} \text{ for } n > 5
	\end{equation*}

	w.h.p., for at most $\mathcal{O}(\log^2{n})$ invocations.
\end{proof}

We provide a tighter bound for the normalization constant, $z'$, in a similar
fashion:

\begin{proof}[of \Cref{lem:rh_norm_tight}]
	Recall from \Crefpart{lem:rh_nested_lattice}{lem:rh_const_upper} that balls
	of radius $2 \sqrt{k}$ have at most 18 highway nodes with probability at
	least $1/2$. When this occurs, $z_{0, \text{inner}}$ can be improved to:
	\begin{equation*}
		z_{0, \text{inner}} < \sum_{j = 1}^{5}{\frac{4j}{j^2}} = 4 \mathcal{H}(5) < 10
	\end{equation*}

	Meanwhile, $z_{0, \text{outer}}$ changes to:
	\begin{equation*}
		z_{0, \text{outer}} < \frac{41 \log{n}}{(2 \sqrt{k})^2}
		= \frac{41}{4} \frac{\log{n}}{k}
	\end{equation*}

	Overall, with probability at least $1/2$, we obtain the improved bounds on
	the normalization constant:
	\begin{equation*}
		z' < 10 + 37 \frac{\log{n}}{k} \text{ for } n > 2
	\end{equation*}
\end{proof}

\subsubsection{Probability of Distance Halving}

As explained before, the first step is to show that we can use the improved
bounds on the normalization constant by incurring only an increase in a constant
factor to the probability of halving the distance:

\begin{proof}[of \Cref{lem:rh_improved_z}]
	The probability of the improved normalization constant bound $z'$ applying
	is at least $1/2$, and this probability is independent for any nodes a
	distance of at least $4 \sqrt{k}$ apart (see
	\Crefpart{lem:rh_nested_lattice}{lem:rh_const_upper}). For values of $k \in
	o\left (\frac{\log{n}}{\log{\log{\log{n}}}}\right )$, the improved
	normalization constant bound is already only a constant factor better. For
	values of $k \in \Omega\left (\frac{\log{n}}{\log{\log{\log{n}}}}\right )$
	we will show that we can always visit highway nodes that are at least $4
	\sqrt{k}$ apart, so that we have independence. All our routing algorithms
	expect to take $\mathcal{O}(\log{n})$ hops on the highway, or $a \log{n}$
	hops for some constant $a$. We expect at least $\frac{1}{2} a \log{n}$ of
	the highway nodes visited to have the improved bounds apply. By Chernoff
	bound, we visit at least $\frac{1}{4} a \log{n}$ highway nodes with the
	improved bounds w.h.p. (with probability at least $1 - n^{-\frac{a}{16
	\ln{2}}}$). Since $a$ can be picked arbitrarily large, then with high
	probability we will visit $\mathcal{O}(\log{n})$-many nodes with the
	improved bounds along our path, which is the same as our original
	expectation of how many nodes we will visit, meaning our results are the
	same up to a constant hidden by the asymptotic notation. Note that a similar
	reasoning works for smaller values of $k$ as well.
\end{proof}

Next, we need to prove a lower bound on how many nodes are in a better phase
than us w.h.p.:

\begin{proof}[of \Cref{lem:rh_ballj}]
	Kleinberg showed that there are more than $2^{2j - 1}$ nodes within lattice
	distance $2^j$ of $t$ \cite{kleinbergSmallworldPhenomenonAlgorithmic2000a},
	for $\log{\log{n}} \leq j < \log{n}$. Within this range, we expect there to
	be at least $2^{2j - 1}/k$ highway nodes. Since we are only considering the
	case where $j \geq \log(c(k + \log{n}))$, we can use this to create a
	Chernoff bound (with $\delta = 1/2$). Letting $X$ be the number of highway
	nodes:
	\begin{align*}
		\Pr(X \leq \mu/2) &\leq e^{-\frac{\mu}{8}}
		= e^{-\frac{2^{2j - 1}}{8k}}
		\leq e^{-\frac{2^{2\log(c(k + \log{n}))}}{16k}} \\
		&= e^{-\frac{[c(k + \log{n})]^2}{16k}}
		< e^{-\frac{c^2(2k\log{n})}{16k}}
		= n^{-\frac{c^2}{8 \ln{2}}} \\
		&< n^{-0.18c^2}
	\end{align*}
	
	In summary, since we picked $\delta = 1/2$, we expect at least $2^{2j -
	2}/k$ highway nodes, to be within lattice distance $2^j$ of $t$ w.h.p. (with
	probability at least $1 - n^{-0.18c^2}$).
\end{proof}

Finally, we use these results to prove the main lemma of this section, the
probability of halving the distance:

\begin{proof}[of \Cref{lem:rh_halving}]
	From our previous results, we know we can use the improved bounds for the
	normalization constant, $z' = 10 + 37\frac{\log{n}}{k}$, with at most a
	constant factor increase in the probability of halving the distance.
	Furthermore, we know that there exist at least $2^{2j - 2}/k$ highway nodes
	in better phases w.h.p. Since they are in phase $j$ or better, they are each
	within lattice distance $< 2^{j + 1} + 2^{j} < 2^{j + 2}$ from $u$. Using
	this, and letting $v$ be an arbitrary long-range connection of $u$, we
	obtain:
	\begin{equation*}
		\Pr(v \in B_{2^j}(t)) > [64 k z']^{-1} > [64k \times 37(1 + \log(n)/k)]^{-1}
	\end{equation*}

	The probability of $v$ not being in a better phase is similarly $1 - \Pr(v
	\in B_{2^j}(u))$. Recalling that each highway node has $Qk$ independently
	chosen random long-range connections, the probability of none of them being
	connected to a better phase is therefore $(1 - \Pr(v \in B_{2^j}(u)))^{Qk}
	\leq e^{-Qk \Pr(v \in B_{2^j}(u))}$. The probability of any one of them
	being connected is therefore:
	\begin{equation*}
		\Pr(\exists v \in B_{2^j}(u)) \geq 1 - e^{-Qk \Pr(v \in B_{2^j}(u))}
		> 1 - e^{-\frac{Qk}{2368(k + \log{n})}}
	\end{equation*}

	When $k \in o(\log{n})$, the $\log{n}$ term in the denominator dominates,
	and we obtain similar asymptotic results to \Cref{lem:kh_halving}. When $k
	\in \Omega(\log{n})$, the $k$ term in the denominator dominates, cancelling
	out the $k$ term in the numerator, and leaving us with a constant term
	dependent on $Q$. It is worth noting that the constant factors in this
	analysis are very loose, and also considerably decrease for larger values of
	of $n$. In any case, we obtain that the probability of halving the distance
	is at least in $\mathcal{O}(k/\log{n})$ for $k \in o(\log{n})$, and at least
	$f(Q) = \mathcal{O}(1)$ for $k \in \Omega(\log{n})$.
\end{proof}

\subsection{Removing Local Contact Dependence}

In this section, we complete the proof of \Cref{thm:rh_routing} by removing the
dependence on local connections. The results of the theorem directly follow.

If we do find a long-range connection that takes us to the next phase, we can
just take it, but what do we do when there aren't any? To continue the Kleinberg
analogy, we would just keep taking local connections to keep re-rolling the
dice, and as long as we never traverse any space twice and never traverse any
space that is within $4 \sqrt{k}$ of previous spaces (because of
\Crefpart{lem:rh_nested_lattice}{lem:rh_const_upper}), we can assume each step
taken is independent of other steps. The obvious problem here is that there is
no notion of ``local connections'' in this randomly selected highway. We could
either greedily take local connections in the entire graph until we happen to
reach a highway node again (in expected $\mathcal{O}(k)$ time), or we can simply
pick any long-range connection that takes us closer to the destination by at
least $4 \sqrt{k}$. For values of $k \in o(\log{n})$, we will use the first
method (greedily taking local connections), and for values of $k \in
\Omega(\log{n})$, we will use the second.

\subsubsection{Values of $k \in o\left
(\frac{\log{n}}{\log{\log{\log{n}}}}\right )$}

For these smaller values of $k$, from \Cref{lem:rh_halving}, we expect to take
$\mathcal{O}(\log(n)/k)$ hops on highway nodes to reach the next phase, and
since there are at most $\log{n}$ total phases, we expect to visit at most
$\mathcal{O}(\log^2(n)/k)$ highway nodes throughout the entire routing process
w.h.p. In the worst case, whenever we can't halve the distance, we never have
any closer long-range connections, so we would need to greedily move along local
contacts towards $t$ until reaching another highway node. Recalling that each
node has probability $1/k$ of being a highway node, and that we expect to visit
a highway node every $k$ independent hops. In order to avoid visiting highway
nodes within $4 \sqrt{k}$ of each other, we can first walk $4 \sqrt{k}$ hops
before checking for highway nodes, which we will expect to find after $4
\sqrt{k} + k \in \mathcal{O}(k)$ hops. Over the entire duration of the routing,
we expect to spend $\mathcal{O}(\log^2(n)/k \times k) = \mathcal{O}(\log^2{n})$
hops using local connections to reach highway nodes w.h.p.

\subsubsection{Values of $k \in \Omega\left
(\frac{\log{n}}{\log{\log{\log{n}}}}\right )$}

For these larger values of $k$, we will prove that we can find a long-range
connection to an arbitrary highway node $u$ in phase $\log(c(k + \log{n})) \leq
j < \log{n}$ that is at least $4 \sqrt{k}$ closer to the destination $t$, w.h.p.
Recall that long-range connections are always only between highway nodes, so
taking them will always keep us on the highway. To find the probability of one
of these connections existing, we consider a ball of radius $d - 4 \sqrt{k}$
centered on the destination $t$ ($B_{d - 4 \sqrt{k}}(t)$), where $d$ is the
distance from $u$ to $t$ ($d = d(u, t)$). Let's lower bound the probability of
an arbitrary long-range connection of $u$ going into this ball. We can assume w.l.o.g.
that $u$ shares either an $x$ or a $y$ coordinate with
$t$ (see \Cref{lem:smaller-ball-overlap}). As before, let's consider the
nested lattice construct, where this time $u$ sits at the edge of one such ball.
There are exactly $2b - 1$ balls closer to $t$ than $u$ is at ball distance $b$,
for $1 \leq b \leq \frac{2d - 2}{6\sqrt{k \log{n}}}$. In order to enforce the
condition that we improve the distance by at least $4 \sqrt{k}$, we can dismiss
the outer layer of balls, leaving us with $2b - 3$ balls for $2 \leq b \leq
\frac{d - 1}{3 \sqrt{k \log{n}}} - 1$. The maximum distance from $u$ to any node
in one of these balls is $2b \times 3 \sqrt{k \log{n}}$. From
\Crefpart{lem:rh_nested_lattice}{lem:rh_n_lower}, we know that each ball of
radius $3 \sqrt{k \log{n}}$ has at least $9 \log{n}$ highway nodes w.h.p.
This lower bound must apply w.h.p. for any highway node along our path, so we
must use the looser normalization constant bound, $z$. We can now lower bound
the probability that $v$ is in one of these closer balls:
\begin{align*}
	\Pr(v \in \mathcal{B}_{d - 4\sqrt{k}}(u) &\geq \sum_{b = 2}^{\frac{d - 1}{3 \sqrt{k \log{n}}} - 1}{
		\frac{(\text{min \# dist $b$ highway nodes})}{z (\text{max dist to node at dist } b)^2}} \\
	&\geq \sum_{b = 2}^{\frac{d - 1}{3 \sqrt{k \log{n}}} - 1}{
		\frac{(2b - 3) \times 9 \log{n}}{z(2b \times 3 \sqrt{k \log{n}})^2}} \\
	&= \frac{2}{9kz} \sum_{b = 2}^{\frac{d - 1}{3 \sqrt{k \log{n}}} - 1}{
		\frac{2b - 3}{b^2}} \\
	&> \frac{2}{9kz} \left [ \ln\left ( \frac{d - 1}{3 \sqrt{k \log{n}}} - 1 \right )  \right ] \\
	&> \frac{\ln\left ( \frac{d}{3 \sqrt{k \log{n}}} \right )}{9kz}
\end{align*}

Note that this result holds for $d \ge c(k + \log{n})$ for large enough constant
$c$.

This result holds for a single long-range connection of $u$. The probability
that none of $u$'s long-range connections are closer is:
\begin{align*}
	\Pr(\text{none closer}) &< \left [1 - \frac{\ln\left ( \frac{d}{3 \sqrt{k \log{n}}} \right )}{9kz} \right ]^{Q k} \\
	&= \left ( \left [ 1 - \frac{\ln\left ( \frac{d}{3 \sqrt{k \log{n}}} \right )}{9kz} \right ]^{k z} \right )^{\frac{Q}{z}} \\
	&< e^{-\frac{Q}{9z} \ln\left ( \frac{d}{3 \sqrt{k \log{n}}} \right ) } \\
	&< e^{-\frac{Q \ln{d}}{9 z}} = d^{-\frac{Q}{9z}}
\end{align*}
again, holding for large enough constant $c$.

With this probability established, let's try seeing how many hops we can take
before we hit a dead end. Let's do this in two parts. First, let's see if we can
get to within a distance of $(a \log{n})^{b z}$ from $t$ for some constants $a$
and $b$. Since the probability of hitting a dead end only increases as we get
closer, the probability of hitting a dead end while in this range is always
going to be $< (a \log{n})^{-\frac{b Q}{9}}$. This gives us an expected number
of hops of $\Omega\left ( (a \log{n})^\frac{b Q}{9} \right )$ w.h.p. When
setting $b$ large enough, we can get this to be $\Omega(\log^2{n})$, which is
more than the maximum number of steps we expect to spend in routing.

In the second part, we are within distance $(a \log{n})^{b z} \geq d \geq c(k +
\log{n})$ of $t$. From \Cref{lem:rh_norm_loose}, we know that our normalization
constant $z$ is at most $\mathcal{O}(\log{\log{\log{n}}})$ for $k \in
\Omega\left (\frac{\log{n}}{\log{\log{\log{n}}}}\right )$ w.h.p., so $z < w
\log{\log{\log{n}}}$ for some constant $w$. This gives us probability of hitting
a dead end of less than $(c (k + \log{n}))^{-\frac{b Q}{9 w \log{\log{\log{n}}}}}$.
Setting constant $c$ large enough, we can expect to take at least $\Omega\left (
\log{n}^\frac{Q}{9 w \log{\log{\log{n}}}} \right )$ hops on the highway within
this range before hitting a dead end. Let's call this our ``allowance''.
While this is less than the maximum number of steps we expect to spend while
routing, we only have at most $b z \log(a \log{n})$ phases left in this second
part, while we spend at most $\mathcal{O}(\log{\log{\log{n}}})$ highway hops per
phase. Putting this together, we expect to take at most
$f(\log{\log{\log{n}}})^2 \log{\log{n}}$ hops in this second part of
the routing for some large enough constant $f$. Let's determine if our allowance
is enough to get us to $t$, by considering the ratio $r$ between our allowance
and the number of remaining highway hops:
\begin{align*}
	r &= \lim_{n \to \infty}{\frac{\log{n}^\frac{Q}{9 w \log{\log{\log{n}}}}}{f(\log{\log{\log{n}}})^2 \log{\log{n}}}} \\
	\log{r} &= \lim_{n \to \infty}{\frac{Q \log{\log{n}}}{9 w \log{\log{\log{n}}}} - \log(f(\log{\log{\log{n}}})^2 \log{\log{n}}) } \\
	&= \lim_{n \to \infty}{\frac{\log{\log{n}}}{\log{\log{\log{n}}}} - \log((\log{\log{n}})^3) } = \infty
\end{align*}

Since $\log{r}$ tends towards infinity, $r$ tends towards infinity, meaning that
for a large enough constant $c$, our allowance is enough to get us to $t$ w.h.p.
for arbitrarily large $n$. Combining these results, we can conclude that we can
reach a highway node within distance $c(k + \log{n})$ of $t$ w.h.p. while only
taking long-range connections that improve our distance by at least $4
\sqrt{k}$, thus eliminating the need for local connections.

\subsection{Randomized Highway Variant}\label{sec:rhvariant}

If it is desired to improve the greedy decentralized routing time of the
randomized highway model for smaller values of $k$ to be in line with the
Kleinberg highway model, it is possible to reintroduce local connections within
the highway nodes, despite the fact that nodes are picked arbitrarily. One
straightforward way to do so is to add a local connection between each highway
node to an arbitrary highway node in each of the 8 adjacent balls of radius
$3\sqrt{k \log{n}}$ (see \Cref{fig:nestedlattice}). From
\Crefpart{lem:rh_nested_lattice}{lem:rh_n_lower} we know that at least one
highway node will exist in each of those balls w.h.p. At least one of these
adjacent highway nodes will be at least $3 \sqrt{k \log{n}}$ closer to the
destination. With this variant, the routing time for smaller values of $k$ is
improved to $\log^2(n)/k$, while only increasing the average degree by a
constant, in line with the randomized highway model. However, this model is not
as clean as the original, and still maintains the same optimal parameter $k$ of
$\Theta(\log{n})$ with the same result of $\Theta(\log{n})$ hops, so we will not
consider it further.

\subsection{Windowed NPA Proofs}\label{sec:wnpaproof}

In this section, we prove that the windowed NPA model maintains a constant average degree
while having a greedy, decentralized routing algorithm taking at most
$\mathcal{O}(\log^{1 + \epsilon}{n})$ hops w.h.p. Specifically, we will define the
routing algorithm as follows: define the subgraph made of nodes with popularity
$\log{n} \leq k \leq A \log{n}$ as the highway, ignoring any long-range
connections that do not connect two ``highway'' nodes. We expect to have
$\mathcal{O}(1 / \log^{1 + \epsilon}{n})$ highway nodes. Using the results from the
previous section, we are able to route in $\mathcal{O}(\log^{1 + \epsilon}{n})$
hops w.h.p.

First, we prove the expected constant average degree:
\begin{lemma}
	The average node degree in the windowed NPA model is $Q$.
\end{lemma}

\begin{proof}
	\begin{equation*}
		\int_{k = 1}^\infty{\epsilon Q k / k^{2+ \epsilon} dk} = \epsilon Q \int_{k = 1}^\infty{1 / k^{1 + \epsilon} dk} = \epsilon Q \times 1 / \epsilon = Q
	\end{equation*}

	Where the normalization constant to pick $k$ is:
	\begin{equation*}
		\int_{k = 1}^\infty{1 / k^{2 + \epsilon} dk} =\frac{1}{1 + \epsilon}
	\end{equation*}
\end{proof}

Next, we show that there are an expected $\mathcal{O}(1 / \log^{1 + \epsilon}{n})$
highway.

\begin{lemma}
	There are $\Theta(\log^{1 + \epsilon}{n})$ highway nodes w.h.p.
\end{lemma}

\begin{proof}
	Now, let's find the probability that a node has popularity between $\log n$ and $A
	\log n$:
	\begin{align*}
		\Pr(\log n \leq k \leq A \log n) &= \int_{k = \log{n}}^{A \log{n}}{\Pr(k) dk} \\
		&= \int_{k = \log{n}}^{A \log{n}}{1/k^{2+\epsilon} dk} \\
		&= \frac{(A^{1 + \epsilon} - 1) \ln^{1+ \epsilon}(2)}{(1 + \epsilon)A^{1 + \epsilon}} \frac{1}{\log^{1 + \epsilon}{n}}
	\end{align*}

	Since $A$ and $\epsilon$ are predetermined constants, the probability that a
	node has a popularity in this range is $\propto \log^{-(1 + \epsilon)}(n)$.
\end{proof}

Importantly, each node within this range of popularities considers all other
points within this range of popularities as long-distance node candidates with
equal likelihoods, a requirement important for the analysis of the randomized
highway model. Next we must prove:
\begin{lemma}
	Each highway node expects to connect a constant fraction of its connections
	to other highway nodes, where the constant is at least $[1 + A^{1 +
	\epsilon}]^{-1}$.
\end{lemma}

\begin{proof}
	The case where there is the least probability of overlap is when $k =
	\log{n}$. Let's consider $v$, an arbitrary long-range connection of node
	$u$, where $k_u = \log{n}$. The probability that $v$ is part of the highway
	is:
	\begin{equation*}
		\Pr(v \in \text{highway}) = \frac{\int_{k = \log{n}}^{A \log{n}}{k^{-2-\epsilon} dk}}{\int_{k = \log{n}/A}^{A \log{n}}{k^{-2-\epsilon} dk}}
		= [1 + A^{1 + \epsilon}]^{-1}
	\end{equation*}
\end{proof}

This is enough to set up an instance of the randomized highway model. An $(N, P,
Q, \epsilon, A)$ instance of the windowed NPA model corresponds with an $(N' =
N, P' = P, Q' = \epsilon Q [1 + A^{1 + \epsilon}]^{-1}, k' = \log^{1 +
\epsilon}{n})$ instance with a few minor modifications. The highway graph,
instead of consisting of nodes with degrees $k$, consists of nodes with degrees
$\log{n} \leq k \leq A \log{n}$.

A little nuance applies since while $k = \log^{1 + \epsilon}{n}$, each of the
nodes has fewer connections, only $\mathcal{O}(\log{n})$. However, the constant
probability of halving the distance analysis still holds, and this algorithm
achieves $\mathcal{O}(\log^{1 + \epsilon}{n})$ expected total greedy-routing
steps. This concludes the proof for \Cref{thm:bak_routing}.

\subsection{Miscellaneous Proofs}

\begin{lemma}\label{lem:smaller-ball-overlap} 
    Let $S_d(w)$ denote the set of vertices at lattice distance $d$ away from any vertex $w$. 
    Let $u$ be any vertex, and let $v$ be any vertex such that $v \in S_d(u)$, and let 
    $B = B_d(u)$. Then $|S_j(v) \cap B|$ is $\Theta{(j)}$ for all $1 \leq j \leq 2d$.
\end{lemma}

\begin{proof}
    Consider the ratio $R_{j, v} = \frac{|S_j(v) \cap B|}{|S_j(v)|}$ at each
	$1 \leq j \leq 2d$. It is clear that no matter where $v$ is located in $S_j(u)
	$, $R_{j, v}$ always grows smaller as $j$ increases. The value of $j$ that
	minimizes $R_{j, v}$ for a particular $v\in S_d(u)$ is then $2d$, and we can
	achieve $\min_v(R_{v, 2d})$ when $v$ is a non-corner vertex in $S_d(u)$, in
	which case $R_{v, 2d}=\frac{d}{8d} = 1/8$. Therefore at every $1 \leq j \leq
	2d$, we have that $\frac{1}{8} \leq \frac{|S_j(v) \cap B|}{4j}$, and
	therefore $|S_j(v) \cap B| \geq j/2$. Since we already have that $|S_j(v)
	\cap B|\leq|S_j(v)| \leq 4j$, the lemma follows.
\end{proof}

\end{document}